\documentclass{llncs}
  
\usepackage{graphicx} 
\usepackage[T1]{fontenc}
\usepackage[utf8x]{inputenc}
\usepackage{csquotes}
\usepackage[english]{babel}
\usepackage{eurosym}
\usepackage{palatino}
\usepackage{tikz}
\usepackage{verbatim}
\usepackage{fancybox}

\usepackage{algorithmic}
\usepackage{algorithm}
\floatname{algorithm}{Procedure}

\usepackage{subfig}
\usepackage{enumerate}
\usepackage{epsfig}
\usepackage[active]{srcltx}
\usepackage{amsfonts}
\usepackage{amsmath}
\usepackage[mathlines]{lineno}
\usepackage{xspace}
\usepackage{tikz}
\usepackage{authblk}
\newcommand{\restate}[3]{
\medskip
\noindent {\bf #1~#2.} 
\textit{#3}
\medskip
}

\newcommand{\restatethm}[2]{\restate{Theorem}{#1}{#2}}
\newcommand{\restatelem}[2]{\restate{Lemma}{#1}{#2}}

\newcommand{\eps}{\epsilon}

\newcommand{\brm}[1]{\operatorname{#1}}

\title{A Near-Optimal Mechanism for Impartial Selection}
\author{Nicolas Bousquet\inst{1}\inst{3} \and Sergey Norin\inst{1} \and Adrian Vetta\inst{1}\inst{2}}
\institute{Department of Mathematics and Statistics, McGill University 
\and School of Computer Science, McGill University 
\and Group for Research in Decision Analysis (GERAD), HEC Montr\'eal}

\begin{document}

\maketitle

\begin{abstract}
We examine strategy-proof elections to select a winner amongst a set of agents, each of whom 
cares only about winning. 
This {\em impartial selection} problem was introduced independently by Holzman and 
Moulin~\cite{HM13}  and Alon et al.~\cite{AFP11}.
Fisher and Klimm~\cite{FK14} showed that the permutation mechanism is impartial and $\frac12$-optimal, 
that is, it selects an agent who gains, in
expectation, at least half the number of votes of most popular agent. Furthermore, they showed the mechanism 
is $\frac{7}{12}$-optimal
if agents cannot abstain in the election.
We show that a better guarantee is possible, provided the most popular agent receives 
at least a large enough, but constant, number of votes. 
Specifically, we prove that, for any $\epsilon>0$, there is a constant $N_{\epsilon}$ 
(independent of the number $n$ of voters) 
such that, if the maximum number of votes of the most popular agent is at least $N_{\epsilon}$ then the
permutation mechanism is $(\frac{3}{4}-\epsilon)$-optimal. This result is tight.\\
Furthermore, in our main result, we prove that near-optimal impartial mechanisms exist. In particular, there is an impartial mechanism
that is $(1-\epsilon)$-optimal, for any $\epsilon>0$, provided that the maximum number of votes 
of the most popular agent is at least a constant $M_{\epsilon}$. 
\end{abstract}

\section{Introduction}
Imagine an election where the voters are the candidates and each voter 
is allowed to vote for as many of the other candidates as she wishes.
Now suppose each voter cares only about winning. 
The goal of the mechanism is to elect the candidate with the maximum support.
To achieve this, we desire that the election mechanism be strategy-proof.
Thus, we want an {\em impartial mechanism}, where voting truthfully
cannot affect an agent's own chances of election.

This problem, called the {\em impartial selection problem}, was introduced independently by Holzman 
and Moulin~\cite{HM13}  and Alon et al.~\cite{AFP11}. In addition to elections, they were motivated
by nomination mechanisms for prestigious prizes and committees, hyperlink formations,
and reputation systems in social networks. Fisher and Klimm \cite{FK14} also proposed the use of 
such mechanisms for peer review evaluation processes.

The impartial selection problem can be formalized via a directed graph $G=(V,A)$.
There is a vertex $v\in V$ for each voter (candidate) $v$, and there is an arc $(u, v)\in A$ 
if $u$ votes for $v$. 
The aim is to maximize the in-degree of the selected vertex, and we
say that an impartial mechanism is $\alpha$-optimal, for $\alpha\le 1$, if the
in-degree of the vertex it selects is always at least $\alpha$ times the
in-degree of the most popular vertex.

Unfortunately, Moulin~\cite{HM13}  and Alon et al.~\cite{AFP11}
observed that a {\em deterministic} impartial mechanisms must have an arbitrarily poor
approximation guarantee $\alpha$. Specifically, a deterministic mechanism may have to select a vertex with zero
in-degree even when other vertices receive votes; it may even be
forced to select a vertex with in-degree one whilst another vertex receives $n-1$ votes!
This negative result motivated Alon et al.~\cite{AFP11} to study {\em randomized} impartial mechanisms.
In particular, they examined a simple mechanism dubbed the {\em 2-partition mechanism}.
This mechanism independently assigns each vertex to one of two groups $\{V_1, V_2\}$.
Then, only the arcs from vertices in $V_1$ to vertices in $V_2$ are counted as votes.
The vertex with the maximum number of {\em counted} votes in $V_2$ is selected (breaking ties arbitrarily).
It is straight-forward to verify that this mechanism is impartial and is
$\frac14$-optimal. They further conjectured the existence of an $\frac12$-optimal impartial 
randomized mechanism.

This conjecture was recently proven by Fisher and Klimm \cite{FK14}. Specifically, they
proved that the {\em permutation mechanism} is impartial and $\frac12$-optimal.
This election mechanism examines the vertices in a random order,
and can only count the votes of a vertex that go to vertices behind it in the 
ordering. (See Section \ref{sec:rand} for a detailed description of the mechanism
and a short proof of Fisher and Klimm's result.)
Interestingly, the factor $\frac12$-approximation guarantee is tight. Consider an $n$-vertex
graph containing only a single arc $(u,v)$. Then, unless $u$ is before $v$ in the random permutation the
mechanism will select a vertex with in-degree zero. Thus the expected
in-degree of the victor is at most one half. 
 
 Observe that this tight example is rather unsatisfactory. It is extremely unnatural and relies on the fact that
every vertex bar one abstains from voting. Indeed, Fisher and Klimm \cite{FK14} showed that
without abstentions the permutation mechanism is at least $\frac{7}{12}$-optimal.
They leave open the possibility that the permutation mechanism actually proffers
a better approximation guarantee than $\frac{7}{12}$. They do prove, however, that without abstentions
the permutation mechanism can be no better than $\frac{2}{3}$-optimal.
Moreover, Fisher and Klimm  \cite{FK14} provide an even stronger inapproximation bound: 
{\bf no} impartial mechanism can be better than $\frac{3}{4}$-optimal, even without abstentions.

This appears to severely limit the potential for progress. But, again, the lower bounds
are somewhat unsatisfactory. The issue now is not low out-degrees (that is, abstentions) but
rather low in-degrees. The lower bounds are all based on instances with extremely small
maximum in-degree $\Delta^-$. Specifically, the factor $\frac12$ optimal example~\cite{AFP11} for the 
permutation mechanism with abstentions 
has $\Delta^-=1$; the factor $\frac{2}{3}$ optimal example  \cite{FK14} for the permutation mechanism without abstentions 
has $\Delta^-=3$;  the factor $\frac{3}{4}$ optimal example \cite{FK14} for any randomized mechanism without abstentions 
has $\Delta^-=2$. Of course, in applications with a large number $n$ of voters, we would
anticipate that the most popular agent receives a moderate number of votes. 
Do these inapproximability bounds still apply for these more realistic settings?
Interestingly, the answer is no, even for cases where the most popular agent receives only
a (large enough) constant number of votes. Specifically, we first prove that the the permutation mechanism
is nearly $\frac34$-optimal in such instances.
\begin{theorem}\label{thm:permutation}
For any $\epsilon>0$, there is a constant $N_{\epsilon}$ such that if $\Delta^- \ge N_{\epsilon}$ then
the permutation mechanism is $(\frac34-\epsilon)$-optimal.
\end{theorem}
This result is tight. We show that the permutation mechanism cannot produce a guarantee better 
than  $\frac34$ regardless of the magnitude of $\Delta^-$.

This result suggests that it may be possible to find a mechanism that beats the $\frac{3}{4}$-inapproximability
bound of \cite{FK14}, even for constant maximum in-degree. This is indeed the case and spectacularly so.
There is an impartial mechanism, which we call the {\em slicing mechanism}, that produces
a near optimal approximation guarantee.
\begin{theorem}\label{thm:slice}
For any $\epsilon>0$, there is a constant $M_{\epsilon}$ such that if $\Delta^- \ge M_{\epsilon}$ then
the slicing mechanism is $(1-\epsilon)$-optimal.
\end{theorem}
The slicing mechanism differs from previous mechanisms in that it adds an initial sampling phase. In this first phase, it samples 
a small fraction of the vertices. It then uses the votes of these vertices to build a non-random ordering of the other vertices.
This specific ordering is exploited in the second phase to elect a vertex with very high expected in-degree.


These results, as in previous works \cite{AFP11,FK14,HM13}, relate to single-winner elections.
Some of the motivating applications, however, involve multiple-winner elections.
We remark that our main result can be generalized 
to multiple-winner elections via small modifications to the mechanism.

\section{The Model}\label{sec:model}
We begin by formalizing the impartial selection problem and introducing 
some necessary notation. 
An election is represented via a directed graph $G=(V, A)$.
The number of vertices of $G$ is denoted by $n$, and each vertex represents
an agent (voter/candidate). An agent can make multiple votes, but cannot vote for herself 
nor vote more than once for any other agent.
Thus, the graph $G$ is loopless and contains no multiple arcs.

A vertex $u$ is 
an \emph{in-neighbor of $v$} if there is an arc $uv\in A$. In this case, we say that $u$ votes for $v$.
Given a subset $Y \subseteq V$ and $v \in V$, the \emph{in-degree of $v$ in $Y$}, denoted by $d^-_Y(v)$, is 
the number of in-neighbors of $v$ 
in $Y$. For simplicity, we denote $d^-_V(v)$, the \emph{in-degree of $v$}, by $d^-(v)$. 
The maximum in-degree of any vertex in $G$ is denoted in by $\Delta^-(G)$, or simply by $\Delta$ 
when there is no ambiguity.

A mechanism is \emph{impartial} if, for every vertex $v$, the probability of selecting $v$ is not modified 
when the out-neighborhood of $v$ is modified. That is, if $v$ changes its votes then this does not affect the 
probability of $v$ being elected. More formally, take any pair of graphs $G$ and $G'$ on the same vertex set $V$.
Let $v$ be a vertex. Then we require that the probability that $v$ is elected in $G$ 
is equal to the probability that $v$ is elected in $G'$, whenever
$N^+_G(u)=N^+_{G'}(u)$ for every $u \neq v$.

Given $1\ge \alpha \ge 0$, an impartial mechanism is \emph{$\alpha$-optimal} if for any graph $G$, the expected 
degree of the winner differs from the maximum degree by a factor of at most $\alpha$, that is,
\[ \frac{\sum_{v \in V} d^-(v) \cdot \brm{Pr}(v\textrm{ is the winner})}{\Delta} \geq \alpha\]
Finally, given an integer $k$, the set $\{1,2,\ldots,k\}$ is denoted by $[k]$.

\section{The Permutation Mechanism}\label{sec:rand}

In this section, we analyze the permutation mechanism of Fisher and Klimm \cite{FK14}. 
This election mechanism examines the vertices in a random order $\{\pi_1, \pi_2,\dots, \pi_n\}$.
At time $t$, the mechanism selects a {\em provisional leader} $y_t$ from the amongst the 
set $\Pi_t=\{\pi_1,\dots \pi_t\}$. At time $t+1$ the mechanism then examines $\pi_{t+1}$. If $\pi_{t+1}$
receives at least as many votes as $y_t$ from $\Pi_t\setminus y_t$ then
$\pi_{t+1}$ is declared the provisional leader $y_{t+1}$. Otherwise $y_{t+1}:=y_t$. The winner
of the election is $y_n$. A formal description of the permutation mechanism is given in 
Procedure \ref{alg:perm-mech}.
\begin{algorithm}
\caption{The Permutation Mechanism}
\label{alg:perm-mech}
\begin{algorithmic}[PERF]
\STATE {\bf Input:} A directed graph $G=(V,A)$. 
\STATE Let $\pi$ be a random permutation of $V=[n]$.
\STATE $y_1\gets \pi_1$;
\FOR{$i=1$ to $n-1$}
   \IF{$d^-_{\Pi_i\setminus \{y_i\}}(\pi_{i+1}) \ge  d^-_{\Pi_i\setminus \{y_i\}}(y_{i})$}
  \STATE $y_{i+1} \gets \pi_{i+1}$
   \ELSE
   \STATE $y_{i+1} \gets y_{i}$ 
   \ENDIF
\ENDFOR
\STATE {\bf output} $y_n$
\end{algorithmic}
\end{algorithm}


Observe that the permutation mechanism is impartial because it has the following property:
the votes of a vertex are only considered after it has been eliminated. Specifically,
the votes of $\pi_t$ are considered at time $\tau > t$ only if $\pi_t$ is not the provisional
leader at time $\tau-1$. But, if $\pi_t$ is not the provisional leader at time $\tau >t$
then it cannot be elected.
This ensures that eliminated agents have no interest to lie, \emph{i.e.} the mechanism is impartial.
Fisher and Klimm \cite{FK14} proved this mechanism is $\frac12$-optimal 
using an intricate analysis based upon viewing the permutation mechanism as
a generalization of the $2$-partition mechanism.
First, we present a simpler proof of their result.
\begin{theorem}\cite{FK14}\label{thm:rand-half}
The permutation mechanism is $\frac12$-optimal.
\end{theorem}
\begin{proof}
Let $v$ be a vertex with maximum in-degree $\Delta$.
Now suppose exactly $j$ of its neighbors
appear before $v$ in the random ordering $\pi$. 
In this case, at the time $v$ is considered it has received
at least $j-1$ valid votes (one of the $j$ votes may not be counted
if it comes from the provisional leader).

Suppose $v$ is now declared the provisional leader.
Then all $j$ of these votes become valid. 
(Indeed, if one of the in-neighbors of $v$ was the provisional leader,
this is no longer the case.)
On the other-hand, suppose $v$ is now declared a loser.
Then, because ties are broken in favor of newly examined
vertices, the provisional leader must already be receiving at least
$j$ valid votes. Thus in either case, the final winner $y_n$ must also
receive at least $j$ valid votes (the in-degree of the provisional leader is non-decreasing).

Now with probability $\frac{1}{\Delta+1}$, exactly $j$ of its neighbors
appear before $v$, for any $0\le j\le \Delta$. Thus, in expectation, 
the winner receives at least 
$\frac{1}{\Delta+1}\cdot \sum_{j=0}^{\Delta} j = \frac12 \Delta$ votes. 
\qed
\end{proof}
As discussed in the introduction, the factor $\frac12$-approximation guarantee in Theorem \ref{thm:rand-half} is tight. 
This tightness is slightly misleading, though. Recall that  the tight example was a graph with just a single arc. 
In general the permutation mechanism is $\frac34$-optimal. Specifically,

\restatethm{\ref{thm:permutation}}{For any $\epsilon>0$, there is a constant $N_{\epsilon}$ such 
that if $\Delta^- \ge N_{\epsilon}$ then
the permutation mechanism is $(\frac34-\epsilon)$-optimal.}

The $\frac34$ bound in Theorem \ref{thm:permutation} is tight
in a very strong sense. There are tight examples for any choice of $\Delta$, no matter how large; see Theorem \ref{thm:tight}.
The proof of Theorem \ref{thm:permutation} has two basic components. The first is the basic observation, 
used above in the proof of Theorem \ref{thm:rand-half},
that the mechanism will perform well if the vertex $v$ of highest in-degree has many neighbors before it in the
random permutation. The second is the observation that if the mechanism does well when $v$ does not 
participate then it will do at least as well when $v$ does participate. In order to be able to apply these
two observations simultaneously, however, we must show random permutations are "well-behaved".
Specifically, we say that a permutation $\pi$ of $[n]$ is \emph{$(\Delta,\eps)$-balanced} if, for every $0 \leq k \leq n$,
$$|[\Delta] \cap \Pi_{k}| \geq \left( \frac k n -\eps \right)\cdot \Delta$$ 
and we want to show that a random permutation is typically balanced. 

To do this, we need the following result, which provides a 
large deviation bound for the size of intersection of two sets of 
fixed cardinalities. 
\begin{lemma}\label{lem:prob} 
For every  $\eps_1>0$, there exists $N_1$ such that for all positive
 integers $N_1<\Delta \leq n$ and $k \leq n$ the following holds. If $X \subseteq [n]$ are chosen uniformly 
 at random subject to $|X|=\Delta$,  then 
\begin{equation}\label{eq:prob1}
\brm{Pr}\left[\left||X \cap [k]| - \frac{k\Delta}{n}\right| \geq \eps_1\cdot \Delta \right] <\eps_1 
\end{equation}
\end{lemma}
\begin{proof} See Appendix. \qed
\end{proof}

\begin{lemma}\label{lem:balanced} For every $0<\eps_2<1$, there exists $N_2$ such that, for all $n \geq \Delta > N_2$, at 
least $(1-\eps_2)\cdot n!$ permutations of $[n]$ are $(\Delta,\eps_2)$-balanced.
\end{lemma}

\begin{proof} Let $N_1$ be chosen to satisfy Lemma~\ref{lem:prob} with $\eps_1 :=\frac{\eps^2_2}{4}$. We 
choose $N_2  \geq \max(N_1,\frac{12}{\eps_2})$.
Let $k_1,k_2,\ldots,k_l \in [n]$ be a collection of integers such that for every $k \in [n]$ there exists $i \in [l]$ 
satisfying $0 \leq k-k_i \leq \eps_2\cdot \frac{n}{3}$. Clearly such a collection can be chosen with 
$l \leq \frac{n}{\lfloor \eps_2\cdot n/3 \rfloor} \leq \frac{4}{\eps_2}$, where the last inequality holds 
as $n \geq \frac{12}{\eps_2}$. Let $\pi$ be 
a permutation of $[n]$ chosen uniformly at random.
By Lemma~\ref{lem:prob} and the choice of $N_2$ we have
\begin{equation}\label{eq:balanced1}
\brm{Pr}\left[|[\Delta] \cap \Pi_{k_i}|\leq \left( \frac{k_i}{n} - \eps_1\right)\cdot \Delta \right] < \eps_1
\end{equation} 
for every $1 \leq i \leq l$. 

We claim that if $|[\Delta] \cap \Pi_{k_i}| \geq \left( \frac{k_i}{n} - \eps_1\right)\cdot \Delta$ 
for every $0 \leq i \leq l$ then $|[\Delta] \cap \Pi_{k}| \geq \left( \frac k n -\eps_2 \right)\cdot \Delta$ for every $0\leq k \leq n$. 
Indeed, given $k$ let $i \in [l]$ satisfy $0 \leq k-k_i \leq \eps_2\cdot \frac{n}{3}$. Then
\begin{eqnarray*}
|[\Delta] \cap \Pi_k| &\geq& |[\Delta] \cap \Pi_{k_i}| \\ 
&\geq& \left( \frac{k_i}{n} - \eps_1\right)\cdot  \Delta \\ 
&\geq& \left( \frac{k}{n} -\frac{\eps_2}{3} - \eps_1 \right)\cdot \Delta \\
&=& \left( \frac{k}{n} -\frac{\eps_2}{3} - \frac{\eps_2^2}{4} \right)\cdot \Delta \\
&\geq& \left( \frac k n -\eps_2 \right)\cdot \Delta
\end{eqnarray*}
as claimed.
By the union bound applied to (\ref{eq:balanced1}) we have
\begin{eqnarray}\label{eq:balanced2}
\brm{Pr}\left[\forall i: i\le l \: :\: |[\Delta] \cap \Pi_{k_i}| \ge \left( \frac{k_i}{n} - \eps_1\right)\cdot \Delta \right] & \ge & 1- l\cdot \eps_1\notag\\
&\geq&  1 - \frac{4}{\eps_2}\cdot \frac{\eps_2^2}{4} \notag \\
&=& 1-\eps_2
\end{eqnarray}
The lemma immediately follows from (\ref{eq:balanced2}) and the claim above.~\qed
\end{proof}

We are now ready to prove Theorem~\ref{thm:permutation}. 

\vskip 5pt \noindent \emph{Proof of Theorem~\ref{thm:permutation}.}
Let $N_2$ be chosen to satisfy Lemma~\ref{lem:balanced} with $\eps_2:=\frac{\eps}{3}.$ We show 
that $N_{\eps}:=\max(N_2,\lceil \frac{6}{\eps_2} \rceil)$ satisfies the theorem.
Let $v$ be a vertex of $G$ with in-degree $\Delta:=\Delta^{-}(G)$. We assume that $V(G)=[n]$, 
where $v$ is vertex $n$, and $[\Delta]$ is the set of neighbors of $v$. For a permutation $\pi$, let $d(\pi)$ denote 
the in-degree of the winner determined by the mechanism.

Let $\pi'$ be a fixed $(\Delta,\eps_2)$-balanced permutation of $[n-1]$. We claim that
\begin{equation}\label{eq:main}
\brm{E}[ d(\pi) \:|\: \pi|_{[n-1]}=\pi'] \ \ \geq \ \ \left(\frac 34 - \eps/2 \right)\cdot \Delta
\end{equation}
Note that the theorem follows from (\ref{eq:main}), as the probability that $\pi|_{[n-1]}$ is not $(\Delta,\eps_2)$-balanced 
is at most $\eps_2$ by the choice of $N_{\eps}$, and thus 
\begin{equation*}
\brm{E}[d(\pi)] \ \ \geq\ \  (1- \eps_2) \cdot \left(\frac 34 - \eps/2 \right)\cdot\Delta 
\ \ \geq\ \  \left( \frac 34 - \eps \right)\cdot \Delta
\end{equation*}
It remains to prove (\ref{eq:main}). Let $w$ be the winner when the permutation mechanism is applied 
to $G \setminus v$ and $\pi'$, and let $x$ be the number of votes $w$ receives from its left 
(\emph{i.e.} from vertices before it in the permutation). 
Let $\pi$ be a permutation  of $[n]$ such that $\pi|_{[n-1]}=\pi'$. 
It is not hard to check that if at least $x+1$ neighbors of $v$ precede $v$ 
in $\pi$ then $v$ wins the election. Moreover, whilst the addition of vertex $v$ can change the winner (and, indeed, produce a 
less popular winner), it cannot decrease the ``left'' degree of any vertex. Thus, the (new) winner has still in-degree at least $x$ 
after the addition of $v$.
So, as $\pi'$ is $(\Delta,\eps_2)$-balanced, we have $|[\Delta] \cap \Pi_{cn}| 
\geq c\Delta -\eps_2\Delta \geq x+1,$  whenever
$c \geq \frac{x+1}{\Delta}+\eps_2.$ Furthermore, $\brm{Pr}[\pi(x)>cn] \geq 1-c$. Thus the probability that $v$ wins the 
election is at least
$1-(x+1)/\Delta-\eps_2$.
It follows that
\begin{eqnarray*}
 \brm{E}[ d(\pi) \:|\: \pi|_{[n-1]}=\pi']  &\geq& \left(\frac{x+1}{\Delta}+\eps_2\right)\cdot x+\left(1-\frac{x+1}{\Delta}-\eps_2 \right)\cdot\Delta  \\
 &\geq& \frac{\Delta^2 -(x+1)\Delta+(x+1)x}{\Delta}-\eps_2\Delta \\ 
 &\geq& \left(\frac 3 4 - \eps_2 -\frac{1}{\Delta} \right)\cdot\Delta+ \frac{(x-\Delta/2)^2}{\Delta} \\
 &\geq& \left(\frac 34 - \frac{\eps}{2} \right)\cdot\Delta
\end{eqnarray*}

\vspace{-1cm}\qed \vskip 15pt

The $\frac34$ bound provided in Theorem~\ref{thm:permutation} is tight for any $\Delta$. 
\begin{theorem}\label{thm:tight}
For every $0<\eps<1/4$ and every $N>0$, there exists a directed graph $G$ such that $\Delta^-(G) \geq N$ and the 
expected degree
of the winner selected by the permutation mechanism is at most $(\frac 3 4 + \eps)\Delta^-(G)$.
\end{theorem}
\begin{proof}
Without loss of generality we assume that $N \geq 1/\eps$.
Let $G'$ be a directed graph such that  $n:=|V(G')| \geq (N+1)(N^2+N+1)\cdot \log \frac{1}{\eps}$, and $d^-(v)=d^{+}(v)=N$ for every $v \in V(G')$. 
Let $G$ be obtained from $G'$ by adding a new vertex $v_0$ and $2N-1$ directed edges from arbitrary vertices in $V(G')$ to $v_0$.
Thus $\Delta:=\Delta-(G)=d^-(v_0)=2N-1$. 

Now, in $G'$ one can greedily construct a set $Z$ of at least $n/(N^2+N+1)$ vertices,
such that no two vertices of $Z$ have common in-neighbors and no two vertices of $Z$ are joined by an edge.
After a vertex $z\in Z$ is chosen, simply remove $z$, the in-neighbors and out-neighbors of $z$, and the
out-neighbors of $z$'s in-neighbors (the inequality is satisfied since the in-neighbors have a common out-neighbor). Then recurse.
Let $\pi$ be a permutation of $V(G)$ chosen uniformly at random. Let $X_v$ denote the event that a vertex $v \in V(G'
)$ is preceded by all of its neighbors in $\pi$. Clearly $\brm{Pr}[X_v]=\frac{1}{N+1}$ for every $v \in V(G')$, and moreover,
by construction of $Z$, the events $\{X_v\}_{v \in Z}$ are mutually independent. Hence 
\begin{eqnarray*}
\brm{Pr}[\cup_{v \in V(G')}X_v] &\geq& 1 - \left(1 - \frac{1}{N+1} \right)^{\frac{n}{N^2+N}} \\ 
&\geq& 1-\left(1 - \frac{1}{N+1} \right)^{(N+1)\cdot \log \frac{1}{\eps}} \\
&\geq& 1 -\eps.
\end{eqnarray*}
Note that if the event $\cup_{v \in V(G')}X_v$ occurs then one of the vertices of $G'$ receives $N$ 
votes in the permutation mechanism. By symmetry the probability that $v_0$ is preceded by at 
most $N-1$ of its neighbors in $\pi$ is equal to $1/2$.
Thus $v_0$ is not selected as a winner with probability at least $1/2-\eps$. We deduce that the 
expected in-degree of the winner is at most
\begin{eqnarray*}
\left(\frac 12 -\eps\right)\cdot \frac{\Delta+1}{2} + \left(\frac 12 +\eps \right)\cdot \Delta
&= &\left( \frac 3 4 + \eps \right)\cdot \Delta - \eps N + \frac{1}{4} \\
&\leq& \left( \frac 3 4 + \eps \right)\cdot \Delta .\hspace{60pt}\qed
\end{eqnarray*} 
\end{proof}

\section{The Slicing Mechanism}

In this section, we present the {\em slicing mechanism} and prove that it
outputs a vertex whose expected in-degree is near optimal. 

\restatethm{\ref{thm:slice}}{For any $\epsilon>0$, there is a constant $M_{\epsilon}$ such that if $\Delta^- \ge M_{\epsilon}$ then
the slicing mechanism is $(1-\epsilon)$-optimal.}

The constant $M_{\epsilon}$ is independent of the number of vertices and is a 
a polynomial function of $\frac{1}{\epsilon}$. We remark that we have made no attempt to
optimize this constant. 

The {\em slicing mechanism} is formalized in Procedure \ref{alg:slice-mech}. 

\begin{algorithm}[t]
\caption{The Slicing Mechanism}
\label{alg:slice-mech}
\begin{algorithmic}[PERF]
\STATE \ 
\STATE {\bf SAMPLING PHASE} 
\STATE  {\scriptsize [Sample]}\ \  Draw a random sample $\mathcal{X}$, where each vertex is sampled with probability $\eps$. 
\FOR{all $v\in V\setminus  \mathcal{X}$}
\STATE {\scriptsize [Estimated-Degree.]}\ \  $d_e(v) \gets \frac{1}{\epsilon}\cdot d_ {\mathcal{X}}^-(v)$ 
\ENDFOR

\STATE \ 
\STATE {\bf SLICING PHASE} 
\STATE  {\scriptsize [Slices]}\ \ Create $\tau= \lceil \frac{1}{\epsilon^2}\rceil$ sets $\{S_1, \dots, S_{\tau}\}$ initialized to empty sets. 
\STATE $\Delta_e \gets \max_{v \in V\setminus  \mathcal{X}}(d_e(v))$
\FOR{all $v\in V\setminus  \mathcal{X}$}
\FOR{$i=1$ to $\tau$}
\IF{$(i-1) \epsilon^2 \cdot \Delta_e \leq d_e(v) \leq i \epsilon^2 \cdot \Delta_e$} 
\STATE $S_i \gets S_i\cup \{v\}$
\ENDIF
\ENDFOR
\ENDFOR

\STATE \ 
\STATE {\bf ELECTION PHASE} 
\STATE {\scriptsize [Revealed Set]}\ \  $\mathcal{R}\gets \mathcal{X}$ 
\STATE {\scriptsize [Provisional Winner]}\ \  $y_0\gets \textrm{argmax}_{u \in V \setminus \mathcal{R}}(d^-_\mathcal{R}(u))$ 
\hfill {\scriptsize [Break ties arbitrarily.]} 
\FOR{$i=1$ to $\tau$}
\FOR{all $v\in S_i\setminus \{y_{i-1}\}$}
\STATE $\mathcal{R}\gets \mathcal{R}\cup \{v\}$ with probability $(1-\epsilon)$.
\ENDFOR
\STATE $y_i'\gets \textrm{argmax}_{u \in V \setminus \mathcal{R}}(d^-_\mathcal{R}(u))$ \hfill {\scriptsize [Break ties arbitrarily.]} 
\STATE $\mathcal{R}\gets (\mathcal{R}\cup S_i\cup \{y_{i-1}\}) \setminus  \{y_{i}'\}$
\STATE  {\scriptsize [Provisional Winner]}\ \  $y_i\gets \textrm{argmax}_{u \in V \setminus \mathcal{R}}(d^-_\mathcal{R}(u))$ 
\hfill {\scriptsize [Break ties arbitrarily.]} 
\STATE $\mathcal{R}\gets (\mathcal{R}\cup \{y_{i}'\}) \setminus  \{y_{i}\}$ 
\ENDFOR
\STATE The elected vertex is $y_{\tau}$.
\end{algorithmic}
\end{algorithm}

This mechanism consists of three parts which  we now informally discuss.
In the first part, the {\em sampling phase}, we independently at random collect a sample $\mathcal{X}$ of  
the vertices. We use arcs incident to $\mathcal{X}$ to estimate the in-degree of
every other vertex in the graph.
In the second part, the {\em slicing phase}, we partition the unsampled vertices into slices, 
where each slice consists of vertices with roughly the same {\em estimated-degree}. 
The third part, the {\em election phase}, selects the winning vertex.
It does this by considering each slice in increasing 
order (of estimated-degrees). After the $i$-th slice is examined the
mechanism selects as {\em provisional leader}, $y_i$, the vertex that has the 
largest number of in-neighbors amongst the set of vertices $\mathcal{R}$ 
that have currently been eliminated. 
The winning vertex is the provisional leader after the final slice has been examined.

We emphasize, again, that the impartiality of the mechanisms
follows from the fact that the votes of a vertex are only revealed when it has been
eliminated, that is, added to $\mathcal{R}$. 
Observe that at any stage we have one provisional leader; if this leader changes when we
examine a slice then the votes of the previous leader are revealed if its slice has already been 
examined.

\subsection{Analysis of the Sampling Phase}

Observe that the sampling phase is used to estimate the in-degree of
each unsampled vertex $v$. Since each in-neighbor of $v$ is sampled in $\mathcal{X}$ with probability $\epsilon$, 
we anticipate that an $\eps$-fraction of the in-neighbors of $v$ are sampled.
Thus, we have an \emph{estimated in-degree}
$d_e(v):=\frac{1}{\epsilon}\cdot d_{\mathcal{X}}^-(v)$, for each vertex $v\in V\setminus \mathcal{X}$. 
It will be important to know how often these estimates are (roughly) accurate. 
In particular, we say that a vertex $u$ is \emph{$\hat{\eps}$-well-estimated} if $|d_e(u)-d(u)| \leq \hat{\eps} d(u)$. 

We will be interested in the case where $\hat{\eps} \ll \eps$. (In particular, we will later select $\hat{\eps} = \frac{\eps^2}{4}$.)
Before analyzing the probability  that a vertex is $\hat{\eps}$-well-estimated, recall the classical Chernoff bound.
\begin{theorem}\label{thm:chernoff}\emph{[Chernoff bound]}\\
 Let $(X_i)_{i \leq n}$ be $n$ independent Bernouilli variables each having probability $p$. Then
 $$\brm{Pr}[|\sum X_i-pn| \geq \delta pn] \leq e^{-\frac{\delta^2 pn}{3}}.\ \ \ \ 
 $$ 
 \end{theorem}

\begin{corollary}\label{coro:chernoff}
 For any vertex $v$ of in-degree at least $\Delta_0=\max(3000,\frac{9\epsilon^2}{\hat{\eps}^4})$, the 
 probability that $v$ is not $\hat{\eps}$-well-estimated is at most $\frac{1}{d(v)^6}$.
\end{corollary}
\begin{proof}
 The proof is an application of Theorem~\ref{thm:chernoff}. For every in-neighbor $u_i$ of $v$, the vertex $u_i$ is 
 sampled with 
 probability $\epsilon$. Denote by $X_i$ the Bernouilli variable corresponding to ``$u_i$ is in $\mathcal{X}$'' which has 
 value $1$ if $u_i \in \mathcal{X}$ 
 and $0$ otherwise. The variables $X_i$ are obviously independent and identically distributed. 
 Note that $\sum X_i = \epsilon\cdot d_e(u)$
  and its expectation is $\epsilon \cdot d(u)$. 
 \begin{eqnarray*} 
 \brm{Pr}\Big(|d_e(u) - d(u)| \geq \frac{\hat{\eps}}{\epsilon} \cdot d(u)\Big) &=& \brm{Pr}\Big(|\epsilon\cdot d_e(u) 
 - \epsilon \cdot d(u)| \geq \hat{\eps} \cdot d(u)\Big) \\
 &\leq& e^{-\frac{\hat{\eps}^2 d(u)}{3\eps}} \\
 &\leq& e^{-\frac{\hat{\eps}^2 \sqrt{\Delta_0}}{3\eps} \cdot \sqrt{d(u)}}  \\
 &\leq & e^{-\sqrt{d(u)}} \\
 &\leq& \frac{1}{d(u)^6}
 \end{eqnarray*}
Here the first inequality is an application of Theorem~\ref{thm:chernoff} with $\delta=\frac{\hat{\eps}}{\eps}$.
The second inequality holds because $d(u)\ge \Delta_0$.
The third inequality follows as $\Delta_0=\frac{9\epsilon^2}{\hat{\eps}^4}$.
Finally, the fourth inequality holds since $\sqrt{d(u)} \geq 6 \log(d(u))$ when $d(u) \geq \Delta_0 \geq 3000$.
\qed
\end{proof}

Let  $\hat{\eps} := \frac{\eps^2}{4}$.
We will be interested in the probability that every vertex of high degree in a local region is $\hat{\eps}$-well-estimated.
Specifically, let $x$ be a vertex of maximum in-degree $\Delta$. Denote by $N^{-k}(x)$ the set of vertices which 
can reach $x$ with an 
oriented path of length at most $k$. For instance, $N^{-1}(x)$ is the in-neighborhood of $x$ plus $x$. 
Applying the union 
bound with Corollary~\ref{coro:chernoff}, we obtain:
\begin{corollary}\label{coro:inneighbors}
Let $x$ be a vertex of in-degree $\Delta$. If $\Delta \geq \Delta_1=\max(\frac{\Delta_0}{\epsilon^2},\frac{3}{\epsilon^{5}})$ 
then, with probability $(1-\epsilon)$, any vertex of $N^{-3}(x)$ of in-degree at least $\epsilon^2\cdot  \Delta$ is $\hat{\eps}$-well-estimated.
\end{corollary}
\begin{proof}
Since $\epsilon^2\cdot  \Delta \geq \Delta_0$, Corollary~\ref{coro:chernoff} ensures that a
vertex of in-degree at least $\epsilon^2\cdot  \Delta$ is not $\hat{\eps}$-well-estimated 
 with probability at most $\frac{1}{(\epsilon^2 \Delta)^6}$. There are at most $1+\Delta+ \Delta^2+\Delta^3 \leq 3\cdot \Delta^3$
 vertices in $N^{-3}(x)$ since $\Delta \geq 2$. 
 The union bound implies that every vertex in $N^{-3}(x)$ with in-degree at least $\epsilon^2 \Delta$ is $\hat{\eps}$-well-estimated 
 with probability at least $(1-\frac{3\Delta^3}{\epsilon^{12} \Delta^6})$. As $\Delta  \geq \frac{3}{\epsilon^{5}}$, 
 the conclusion holds.~\qed 
\end{proof}

It the rest of this section we will make a set of assumptions.
Given these assumptions, we will prove that the mechanism outputs a vertex of
high expected in-degree. We will say that the mechanism ``fails" if these assumptions do not hold. 
We will then show that the probability that the mechanism fails is very small.
The two assumptions we make are:\\
\vskip 3pt
(A1) Vertex $x$ is not sampled. This assumption fails
with probability~$\eps$. \\ 
\vskip 3pt
(A2) Every
vertex in $N^{-2}(x)$ with in-degree at least  $\epsilon^2 \cdot \Delta$ is regionally well-estimated.
Here, we say a vertex is \emph{regionally well-estimated} if its degree is $\hat{\eps}$-well-estimated and 
all its in-neighbors of in-degree at least $\epsilon^2 \Delta$ 
are also $\hat{\eps}$-well-estimated. Corollary~\ref{coro:inneighbors} ensures that all the vertices of $N^{-2}(x)$ of degree at least 
 $\epsilon^2 \cdot \Delta$ are  regionally well-estimated with probability $(1-\eps)$. Thus, this assumption also fails
 with probability at most $\eps$.

\subsection{Analysis of the Slicing Phase}

Now we consider the slicing phase.
In this phase we partition the unsampled vertices into groups (slices)
according to their estimated degrees.
The {\em width} of a slice is the difference between the upper and lower
estimated-degree requirements for vertices in that group.
We will need the following bounds on the width of a slice.

\begin{lemma}\label{lem:width}
 The width of any slice is at least $(1-\hat{\eps})\epsilon^2 \cdot \Delta$ and at most $\epsilon\cdot \Delta$.
\end{lemma}
\begin{proof}
By assumption, the vertex $x$ of maximum degree is $\hat{\eps}$-well-estimated. 
Thus, $\Delta_e \geq (1-\hat{\eps}) \cdot \Delta$. Therefore the width of any slice is 
at least $(1-\hat{\eps})\epsilon^2 \cdot \Delta$. 

On the other-hand, take any vertex $u$. At at most $\Delta$ of $u$'s in-neighbors can be sampled
because it has degree at most $\Delta$. It follows that $d_e(u) \leq \frac{\Delta}{\epsilon}$.
Thus, $\Delta_e \leq \frac{\Delta}{\epsilon}$, and the width of any slice is at most
$\eps\cdot \Delta$.
 \qed
\end{proof}

\subsection{Analysis of the Election Phase}
We are now ready to analyze the election phase.
Initially we reveal every vertex in the sample $\mathcal{X}$. The vertex $y_0$ with largest estimated-degree is
then the provisional winner. We then treat the slices in increasing order of estimated-degree.
When we consider slice $S_i$, we will reveal every vertex in $S_i$ except one
(if it is the provisional winner $y_i$). For technical reasons, we will denote by $S_0$ the set $\mathcal{X}$. 
Observe that the set $\mathcal{R}$ is the set of already revealed (eliminated) vertices.

Now let $S_{\le \ell}=\cup_{j=0}^{\ell} S_j$, and denote by $d_\ell(u)=|\{u\in S_{\le \ell}: vu\in A\}$ 
the number of in-neighbors of $v$ that are in $S_j$, for $j \leq \ell$. 
Then we begin by proving two lemmas. The first, Lemma~\ref{lemma:smalldegrees}, states that if a vertex $u \in S_{\ell}$ has 
a large $d_{\ell-1}(u)$ then the elected vertex has large in-degree. 
The second, Lemma~\ref{lemma:bigdegrees}, guarantees that the elected vertex has a large in-degree (with high probability)
if there are many  regionally well-estimated vertices in $S_\ell$ with large $d_{\ell}(u)$.
These lemmas will be applied to a vertex $x$ of in-degree $\Delta$: either many in-neighbors of $x$
are in slices before $x$ and Lemma~\ref{lemma:smalldegrees} will apply, or many in-neighbors of $x$ are
in its slice and we will apply Lemma~\ref{lemma:bigdegrees} to this set of in-neighbors.

\begin{lemma}\label{lemma:smalldegrees}
Take $u \in S_{\ell+1}$. If $d_{\ell}(u) = d$, the elected vertex has in-degree at least $d-1$.
\end{lemma}
\begin{proof}
When we select the provisional winner $y_{\ell}$ all vertices of $S_{\le \ell}=\cup_{j=0}^{\ell} S_j$
(but at most one, $y_{\ell}'$, if it is in this set) have been revealed.
Now $u \in S_{\ell+1}$ is an eligible candidate for $y_{\ell}$. Thus, at that time,
$d^-_{\mathcal{R}}(y_{\ell}) \geq d^-_\mathcal{R}(u)\ge d_{\ell}(u)  -1 = d-1$. 
 Since the in-degrees of the provisional winners can only increase, the elected vertex $y_{\tau}$ has is in-degree 
 at least $d-1$. Note that the minus one comes from the fact that $y_\ell'$ can be an in-neighbor of $u$. \qed
\end{proof}

\begin{lemma}\label{lemma:bigdegrees}
Let $\Delta \geq \Delta_1=\frac{12^2}{\hat{\eps}^4}$. 
 If there exists an integer $\ell$ and a set $Z \subseteq S_\ell$ of size at least $\epsilon \Delta$ 
 of regionally well-estimated vertices with $d_{\ell}(z) \geq (1-3\epsilon) \Delta+1$ for every $z \in Z$,
 then with probability at least $(1- \eps)$ we have $d(y_\ell') \geq (1-5\epsilon)\Delta$.
\end{lemma}
\begin{proof}
First, by selecting a subset of $Z$ if necessary, we may assume that $Z = \lceil \eps \Delta \rceil$.
Now we define a collection of bad events and show that $d(y_\ell) \geq (1-5\epsilon)\Delta$ if none
of these events arise. We then show the probability that any of these bad events occurs is small. 

Let ${\bf B}_0$ be the event that every vertex in $Z$ is placed in $\mathcal{R}$ when we
sample the vertices of slice $\ell$. 
We may assume the provisional leader $y_{\ell-1}$ is not in $Z$. Thus, since $|Z| \geq \epsilon \Delta$, 
the probability of event ${\bf B}_0$ is at most
   \begin{equation*} 
 (1-\epsilon)^{\epsilon \Delta} \ \ =\ \  e^{\epsilon\Delta\cdot \log (1-\epsilon)}  
\  \ \leq\  \  e^{-\eps^2 \Delta} 
 \ \ \leq\ \   \frac{\eps}{2}
 \end{equation*} 
 Here the first inequality holds since $\log(1-\eps) \leq -\eps$.
The second inequality holds as $\Delta \geq \frac{2}{\epsilon^4} \ge \frac{12^2}{\hat{\eps}^4}$.
 
 Now take any $z \in Z$ and let $U_z$ be the set of in-neighbors of $z$ in $S_{\leq \ell}\setminus \{y_{\ell-1}\}$. 
 We have $|U_z| \geq (1-3\epsilon)\cdot \Delta$.
 Let ${\bf B}_z$ be the event that less than $(1-5\epsilon)\cdot\Delta$ vertices of $|U_z|$
 are in $\mathcal{R}$ at the time we sample the vertices of slice $\ell$. 
 
 To analyze the probability of this event consider any $u_i \in U_z$.
 Now, if $u_i \in S_j$ for $j< \ell$ then $u _i$ is already in $\mathcal{R}$.
 Otherwise, if $u_i \in S_{\ell}$ then it is now added to $\mathcal{R}$
 with probability $(1-\epsilon)$.
 So consider a random variable $Y_i$ which has value $1$ if $u_i\in \mathcal{R}$ after this 
 sampling and $0$ otherwise. Note the $Y_i$ are not identically distributed.
So let $X_i$ be variables which are independent and identically 
 distributed and such that $X_i=Y_i$ if $u_i \in S_{\ell}$ 
 and $X_i$ has value $1$ with probability $(1-\epsilon)$ otherwise. We have
 \begin{eqnarray*} 
\brm{Pr} [\sum_i Y_i \leq (1-\epsilon-\hat{\eps})\cdot |U_z|] &\leq&
\brm{Pr}[ \sum_i X_i \leq (1-\epsilon-\hat{\eps})\cdot |U_z|] \\ 
&\leq& e^{-\hat{\epsilon}^2 (1-\epsilon)\cdot |U_z|/3} \\
&\leq& e^{-\frac13 \hat{\eps}^2 (1-\epsilon)(1-3\epsilon)\Delta} \\
&\leq& e^{-2 \sqrt{\Delta}} \\
&\leq& \frac{1}{3\Delta}
\end{eqnarray*}
Here the first inequality follows from Theorem~\ref{thm:chernoff}.
The second inequality holds as $|U_z| \geq (1-3\epsilon) \Delta$.
The third inequality holds by the choice $\sqrt{\Delta} \ge \frac{12}{\hat{\eps}^2} \ge \frac{6}{\hat{\eps}^2 (1-\epsilon)(1-3\epsilon)}$.
The fourth one is satisfied since $\Delta \geq 100$.

We now apply the union bound to the events ${\bf B}_0 \cup \bigcup_{z\in Z} {\bf B}_z$.
Since $Z = \lceil \eps \Delta \rceil$, none of these events occur with probability at least 
$1-\frac{\epsilon \Delta+1}{3\Delta}-\frac{\eps}{2}\ge 1-\eps$.    
Thus, with probability at least $1-\eps$, after the sampling of the slice $S_\ell$, there is a vertex $z\in Z$ that is not 
in $\mathcal{R}$ but that has at least $(1-5\epsilon)\cdot \Delta$ in-neighbors in $\mathcal{R}$.
 The new provisional leader $y_{\ell}'$ must then
satisfy $d_{\mathcal{R}}^-(y_{\ell}') \geq d_{\mathcal{R}}^-(z) \geq (1-5\epsilon)\Delta$, as required, since
 $(1-5\epsilon)\cdot\Delta \le  (1-\epsilon-\hat{\eps})\cdot |U_z|$. \qed
 \end{proof}

\vskip 5pt \noindent \emph{Proof of Theorem~\ref{thm:slice}.}
We may now prove that the slicing mechanism is nearly optimal.
We assume that $\Delta \geq M_\epsilon=\max(\Delta_1,\Delta_2)$ 
and that $\epsilon\leq \frac{1}{8}$.
Let $x$ be a vertex of in-degree $\Delta$. We assume that $x$ is not selected in $\mathcal{X}$ during the sampling phase 
and that all the vertices of $N^{-2}(x)$ with in-degree at least $\epsilon^2 \Delta$ are regionally well-evaluated.
We need the following claim, where $k$ denotes the integer such that 
the vertex $x$ of maximum in-degree is in $S_k$. 
\begin{claim}\label{cl:kplusone}
Let $u$ be a vertex in $N^{-3}(x)$ that is not in $S_{\le k}$. Then $u\in S_{k+1}$.
\end{claim}
\begin{proof}
Take a vertex $u\in N^{-3}(x)$. If $d^-(u) \leq \eps^2 \Delta$, then its estimated degree is at most $\epsilon \Delta \leq d_e(x)$. 
Thus now we can assume that $d^-(u) \geq \eps^2 \Delta$ and then $u$ is $\hat{\eps}$-well-estimated by assumption on $x$.
First observe that the set of possible estimated degrees of $u$ intersects at most two slices. 
To see this note that the range of $d_e(u)$ is less than $2\hat{\eps}\cdot \Delta$ as 
$u$ is $\hat{\eps}$-well-estimated. On the other-hand, by Lemma \ref{lem:width},
the width of a slice is at least $(1-\hat{\eps}) \epsilon^2 \cdot \Delta$. 
Since $\hat{\eps} = \frac{\epsilon^2}{4}$ we have 
\begin{equation*}
2\hat{\eps} \Delta \ \ \leq \ \ \frac{1}{2} \epsilon^2 \Delta \ \ <\ \  (1-\hat{\eps}) \epsilon^2 \Delta
\end{equation*}
Since the range is less than the width, the observation follows.
 
The vertex $x$ of maximum degree is $\hat{\eps}$-well-estimated and is in the slice $k$. 
Therefore, because $u$ is $\hat{\eps}$-well-estimated
(and necessarily $d(u)\le d(x)$), there must be a slice smaller than or equal to $k$ in its range of $u$.
Thus $u$ cannot be in a slice with index exceeding $k+1$. \qed
\end{proof}

Assume first that there exists a set $Z_1$ of at least $\epsilon \Delta$ vertices of $N^{-2}(x)$ such that  
$Z_1 \cap S_i = \emptyset$ for $i \leq k$. 
The claim ensures that $Z_1 \subseteq S_{k+1}$. By considering a subset of $Z_1$, we may assume 
that $|Z_1|=\epsilon \Delta$. (We assume that $\eps \Delta$ is an integer, for simplicity.)
Then, for any $z \in Z_1$, we have $d_e(z) \geq d_e(x)$ since $z$ is in the slice after $x$. Moreover, as  
$x$ and $z$ are $\hat{\eps}$-well-estimated, the facts that $|d_e(z) -d(z)| \leq \hat{\eps} \Delta$ 
and $|d_e(x) -d(x)| \leq \hat{\eps} \Delta$ imply that
\begin{equation*}
d(z) \ \ \geq\ \ (1-2\hat{\eps})\Delta \ \ \geq \ \ (1-\epsilon)\Delta
\end{equation*}
Furthermore, by the above claim, we must have that $d_{k+1}(z) = d(z)$. 
Consequently, we may apply Lemma~\ref{lemma:bigdegrees} 
to the set $Z_1$. This ensures that the in-degree 
of the elected vertex is at least $(1-5\epsilon)\Delta$ with probability at least $(1-\frac{1}{\eps})$.

On the other hand, assume that now less than $\epsilon \Delta$ vertices of $N^{-2}(x)$ are in $S_{k+1}$. In particular,
we have $d_k(x) \geq (1-\epsilon)\Delta$. If $d_{k-1}(x) \geq (1-4\epsilon)\cdot \Delta+1$ then 
the elected vertex has degree at least $(1-4\epsilon) \Delta$ by Lemma~\ref{lemma:smalldegrees}. 

So, assume that $d_{k-1}(x) \leq (1-4\epsilon)\Delta$.
Then at least $3\epsilon \Delta$ in-neighbors of $x$ are in $S_k$. 
Denote by $Z_2$ a set of  $\epsilon \Delta$ in-neighbors 
of $x$ in $S_k$. Every vertex $z\in Z_2$  
has in-degree at least $(1-\epsilon-2 \hat{\eps})\Delta \geq (1-2\epsilon)\Delta+1$; this follows because both $x$ and $z$ 
are $\hat{\eps}$-well-estimated and because the width
of a slice is at most $\epsilon \Delta$ (Lemma~\ref{lem:width}). 
For any $z \in Z_2$, since at most $\epsilon \Delta$ of the in-neighbors of $z$ are in $S_{k+1}$,
we have $d_k(z) \geq (1-3\epsilon) \Delta+1$.  Moreover, by assumption all the vertices of $Z_2$ are regionally well-estimated. 
Hence, by Lemma~\ref{lemma:bigdegrees}, with probability at least $(1-\epsilon)$, the degree of the 
elected vertex is at least $(1-5\epsilon)\Delta$, as desired. \vspace{8pt}

Consequently, the slicing mechanism typically outputs a near-optimal vertex. So what is the probability
that assumptions made during the proof fail to hold? Recall that Assumptions (A1) and (A2) fail to hold with probability
at most $2\eps$. Given these two assumptions,  Lemma~\ref{lemma:bigdegrees} fails to 
output a provisional leader with in-degree at least $(1-5\epsilon)\cdot \Delta$ 
with probability at most $\eps$. Thus the total failure probability is at most $3\eps$.
Consequently, the expected in-degree of the elected vertex is at least $(1-3\epsilon)(1-5\epsilon)\cdot \Delta \geq (1-8\epsilon)\Delta$, 
which concludes the proof of Theorem~\ref{thm:slice}. 
 \qed
\vskip 5pt

We conclude with some remarks. Here $M_\epsilon = \mathcal{O}(\frac{1}{\epsilon^8})$; the degree of this polynomial 
can certainly be improved as we did not attempt to optimize it. 

The slicing mechanism can be adapted to select a fixed number $c$ of winners rather than one. 
Let us briefly explain how.
Instead of selecting only one provisional winner $y$ during each iteration of the election phase, we can select a set of size $c$
containing unrevealed vertices maximizing $d^-_ \mathcal{R}$. \\
Let $x_1,\ldots,x_c$ be the $c$ vertices of highest in-degree. With high probability all the vertices of $N^{-2}(x_i)$ are regionally well-estimated
for every $i \leq c$ and with high probability none of them are selected during the sampling phase. 
Now consider two cases: either $N^{-2}(x_1)$ contains many vertices of degree almost $\Delta$, and then an adaptation of
Lemma~\ref{lemma:bigdegrees} ensures that with high probability  $c$ vertices  are not sampled during the sampling of the election phase and
the $c$ elected vertices have large degree. Or $N^{-2}(x_1)$ has few vertices of degree almost $\Delta$
and then when the slice of $x_1$ is considered, if $x_1$ is not selected, then all the selected vertices have degree almost $\Delta$ by Lemma~\ref{lemma:smalldegrees}. 
A similar argument can be repeated for every vertex $x_i$.


 \ \\
 \noindent
{\bf Acknowledgements.} The authors are extremely grateful to Felix Fischer
for introducing us to the impartial selection problem and for discussions.

\section*{Appendix}

In this appendix, we include the proof omitted from the main text due to space constraints.

\restatelem{\ref{lem:prob}}{For every  $\eps_1>0$ there exists $N_1$ such that for all positive
 integers $N_1<\Delta \leq n$ and $k \leq n$ the following holds. If $X \subseteq [n]$ are chosen uniformly 
 at random subject to $|X|=\Delta$,  then 
$\brm{Pr}\left[\left||X \cap [k]| - \frac{k\Delta}{n}\right| \geq \eps_1 \Delta \right] <\eps_1$.}

\begin{proof}  Let $k \leq n$. Clearly, the lemma holds for $\Delta=n$, and so we assume $\Delta < n $. Let $p=\Delta/n$, and 
let $Y$ be a random subset of $[n]$ obtained by choosing every element of $[n]$ independently with 
probability $p$. Our first goal is to lower bound $\brm{Pr}[|Y|=\Delta]$. We start by deriving the following estimate:
\begin{equation}\label{eq:probInduction}
\frac{n!}{(n-\Delta)!}(n-\Delta)^{n-\Delta} \ \  \geq \ \  n^ne^{-\Delta}
\end{equation}
which holds for all non-negative integers $\Delta$. 
We prove (\ref{eq:probInduction}) by induction on $\Delta$. The base case $\Delta=0$ is trivial for any non-negative $n$. 
For the induction step, we have
\begin{eqnarray*}
\frac{n!}{(n-\Delta)!}(n-\Delta)^{n-\Delta} 
&=& n \frac{(n-1)!}{(n-\Delta)!}(n-\Delta)^{n-\Delta} \\
&\geq& n(n-1)^{n-1}e^{-(\Delta-1)} \\
&\geq&  n(n-1)^{n-1}\left(1+\frac{1}{n-1} \right)^{n-1}e^{-\Delta}\\
&=& n^ne^{-\Delta}
\end{eqnarray*} 
Here the first inequality follows by induction, the second inequality
follows from the fact that $(1+\frac{1}{n})^n \leq e$ for every $n$.
Thus (\ref{eq:probInduction}) holds for all $\Delta$.
We then have
\begin{eqnarray*}
\brm{Pr}[|Y|=\Delta] &=& \binom{n}{\Delta}p^{\Delta}(1-p)^{n-\Delta}\\
& =& \frac{n!}{(n-\Delta)!\Delta!}\left(\frac{\Delta}{n}\right)^\Delta \left(1-\frac{\Delta}{n}\right)^{n-\Delta}\\
&=&\frac{\Delta^{\Delta}}{\Delta!} \left(\frac{n!}{(n-\Delta)!}\frac{(n-\Delta)^{n-\Delta}}{n^n}\right) \\
&\geq& \frac{\Delta^{\Delta}}{\Delta!} e^{-\Delta}\\
&\geq& e^{-\Delta}\frac{\Delta^{\Delta}}{3\sqrt{\Delta}(\Delta/e)^\Delta}\\
&=& \frac{1}{3\sqrt{\Delta}}
\end{eqnarray*} 
Here the first inequality follows by (\ref{eq:probInduction}). The second inequality applies
when $\Delta$ is sufficiently large. Indeed Stirling's formula ensures that $n! \sim \frac{1}{\sqrt{2\pi n}}(\frac{n}{e})^n$;
thus, since $\sqrt{2\pi} \leq 3$, we have $\Delta! \leq 3\sqrt{\Delta}(\frac{\Delta}{e})^\Delta$.
Therefore
\begin{equation}\label{eq:prob2} \brm{Pr}\left[\left||X \cap [k]| - \frac{k\Delta}{n}\right| \geq \eps_1 \Delta \right] 
\ \ \leq\ \  3\sqrt{\Delta}\cdot \brm{Pr}\left[\left||Y \cap [k]| - \frac{k\Delta}{n}\right| \geq \eps_1 \Delta \right]
\end{equation}
By Theorem~\ref{thm:chernoff}, applied with $\delta=\eps_1$ and $\Delta=pn$,  we have
\begin{equation}
\label{eq:prob3}
\brm{Pr}\left[\left||Y \cap [k]| - \frac{k\Delta}{n}\right| \geq \eps_1 \Delta \right] \ \ \leq \ \  e^{-\frac{\eps_1^2\Delta}{3}}
\end{equation} 
Combining (\ref{eq:prob2}) and (\ref{eq:prob3}), we deduce that
\begin{equation}
\label{eq:prob4}\brm{Pr}\left[\left||X \cap [k]| - \frac{k\Delta}{n}\right| \geq \eps_1 \Delta \right] \ \  \leq \ \ 3\sqrt{\Delta}e^{-\frac{\eps_1^2\Delta}{3}}
\end{equation}
Clearly, for $\Delta$ sufficiently large with respect to $\eps_1$ we have  $3\sqrt{\Delta}e^{-\frac{\eps_1^2\Delta}{3}} \leq \eps_1$. 
It follows that (\ref{eq:prob4}) gives the lemma for such $\Delta$, as desired.~\qed
\end{proof}

\end{document}